\documentclass{paper}

\usepackage{amsthm}
\usepackage{hyperref}
\usepackage[hyphenbreaks]{breakurl}
\usepackage{float}
\newtheorem*{defn}{Definition}
\newtheorem*{mydef}{Theorem}
\begin{document}\sloppy

\title{Stochastic Contracts for Runtime Checking of Component-based Real-time Systems}
%
%
%
%
%

\numberofauthors{4} 
%
\author{
%
%
\alignauthor
Chandrakana Nandi\\
	   \affaddr{Department of Computer Science, ETH Zurich}\\
       \email{chandrakana.nandi@inf.ethz.ch}\\
\and
\hspace{.2cm}
\alignauthor
Aurelien Monot\\
       \affaddr{ABB Corporate Research}\\
       \email{aurelien.monot@ch.abb.com} \\
\and 
\\
\alignauthor Manuel Oriol\\
     \affaddr{ABB Corporate Research}\\
       \email{manuel.oriol@ch.abb.com}
}
\toappear{}
\maketitle
\begin{abstract}
This paper introduces a new technique for dynamic verification of component-based real-time systems based on statistical inference. Verifying such systems requires checking two types of properties: functional and real-time. For functional properties, a standard approach for ensuring correctness is Design by Contract: annotating programs with executable pre- and postconditions. We extend contracts for specifying real-time properties. 

In the industry, components are often bought from vendors and meant to be used off-the-shelf which makes it very difficult to determine their execution times and express related properties.\ We present a solution to this problem by using statistical inference for estimating the properties. The contract framework allows application developers to express contracts like ``the execution time of component $X$ lies within $\gamma$ standard deviations from the mean execution time''. 
Experiments based on industrial case studies show that this framework can be smoothly integrated into existing control applications, thereby increasing their reliability while having an acceptable execution time overhead (less than 10\%).  
\end{abstract}



\section{Introduction}
A real-time system is one for which respecting deadlines is a part of the system specification. These systems are widely used in developing safety critical applications~\cite{Reference57} such as control systems for nuclear power plants, safety systems of automobiles, and avionics. Verifying such systems ensures that their behaviors satisfy the specifications. Among others, component-based engineering is an approach for developing real-time safety critical systems~\cite{Reference53}, \cite{Reference54}, \cite{Reference55} and a lot of research has been done in verifying these systems~\cite{Reference58}, \cite{Reference59}, \cite{Reference61}, \cite{Reference62}, \cite{Reference63}. Most of these works are based on model checking. Another approach which is better suited for verifying component-based systems due to its support for modular abstraction is \textit{Design by Contract}~\cite{Reference6} (annotating programs with executable pre- and postconditions) for specifying functional and real-time properties~\cite{Reference64}, \cite{Reference66}, \cite{Reference67}, \cite{Reference68} which can be verified using static or dynamic techniques. 

A major challenge in specifying contracts for component-based real-time systems lies in expressing the real-time properties.\ The existing work in this area focuses mainly on checking very simple timed properties which may not be sufficient for representing the behavior of a system. In practice, real-time contracts for component-based systems should be able to account for all real-time aspects such as availability of resources like CPU and memory, exception handling during the execution of a component and interactions with other components. Another problem lies in reasoning about the Worst Case Execution Times (WCETs) of the components. Static analysis is a possible approach but it is not easy to design and implement a static analyzer that can cater to the needs of a particular industrial component-based system. Even if one does so, it would not be reusable for other such systems. Moreover, static analysis may not be feasible when the components are bought from vendors and the source code is not available~\cite{Reference46}. Another approach is to use probabilistic analysis for determining the probability density function of the execution times. This is not applicable to all systems because the execution time distribution depends on the behavior of a particular component. The results for one component can not be generalized to all components. 

In this paper we present solutions to the problems of expressing and checking real-time properties by incorporating empirical stochastic analysis in the contracts: instead of determining the exact WCET values of the components or predicting the execution time probability distributions, we \textit{estimate} the WCET values based on empirical distribution functions and then specify the contracts. Our contract framework allows us to specify properties like `` \textbf{Execution time of component $X$ should lie within $\gamma$ standard deviations from the mean execution time of $X$}'' and check them dynamically. The mean and the standard deviation of the execution times are estimated by statistical inference and the parameter $\gamma$ is a property of a particular component. 

We validated our approach on a state-of-the-art component-based framework, FASA (Future Automation System Architecture)~\cite{Reference1} developed at ABB Corporate Research. 
Our experiments are based on industrial case studies and show that for applications deployed on a single host controller, the average execution time overhead added due to the contracts is less than 10\%, which makes it efficient and easy to incorporate on top of any component-based real-time application. \\

\section{Overview}
\subsection{Motivating Example}
\label{subsec:mtvtn}
We consider a typical control loop with a combination of a feed-forward and a feedback controller as our running example. Figure~\ref{fig:runningeg} shows the schematic diagram for this system.  
\begin{figure}[htbp]
\centering 
\includegraphics[width=.6\columnwidth]{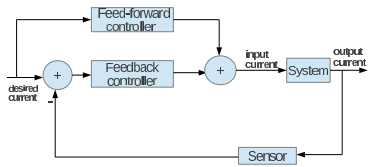} 
\caption[A control system]{A control system with a combination of an open loop and a closed loop controller}
\label{fig:runningeg} 
\end{figure} 
The application has a component, Sensor which measures the value of a physical quantity, in our case electric current generated by System. The feedback loop in our example is based on negative feedback because the output current measured by Sensor is subtracted from the desired current value to generate an error. This error is used by the component, Feedback controller to dynamically regulate the output of the System. The other component is a Feed-forward controller. This controller does not measure the output current generated by System and as a result, the latter has no effect on it. The net effect on the input to System is the sum of the effects of the two controllers. In practice, a combination of a feed-forward and a feedback controller is often used for ensuring faster convergence to the desired value of the physical quantity being measured. 

\subsection{FASA}
FASA (Future Automation System Architecture) is a component based framework developed at ABB Corporate Research~\cite{Reference1} used for \textit{cyclic} real time control applications. A FASA application is made up of \textit{function blocks} (components) which communicate with each other through \textit{ports}. Ports can be input or output. A function block receives data at an input port and sends out data through an output port. An input port is connected to an output port through a unidirectional \textit{channel}. FASA has a centralized \textit{scheduler} which is responsible for scheduling the function blocks on the available controllers. The execution of an application is triggered by the launching of the FASA \textit{kernel} which is the main entry point to the framework. The function blocks are implemented as C++ classes and have a dedicated method, \texttt{operator()} where their behavior is described. 

Referring to our running example, the Feed-forward controller, the Feedback controller, the Sensor, the two Addition Components (shown in Figure~\ref{fig:runningeg} with $+$ sign) and the System are FASA function blocks. The arrows indicate the direction of data flow through the channels.

\subsection{Types of Contracts}
Our contract framework facilitates the specification and monitoring of functional and real-time properties. For the functional specifications, it supports preconditions, postconditions and class invariants. It supports loop invariants and loop variants for checking the correctness of loops and the \texttt{old} construct as used in some languages such as Eiffel~\cite{Reference6}, which specifies properties about the state of an object. 
Figure~\ref{fig:funct_sens} is a code snippet showing a functional contract for Sensor from our running example. Sensor has an attribute \texttt{interval} which records the current cycle number. In every cycle, this value is incremented by one. Before every update, the current value of \texttt{interval} is saved using the \texttt{OLD} macro. In the postcondition, the old value is retrieved using another macro, \texttt{GET_OLD} which takes two arguments, the type and the attribute name. The postcondition states that the value must be correctly updated in each cycle.

\begin{figure} [htb]
\centering 
\includegraphics[trim={0 20cm 0 0},width=1\columnwidth]{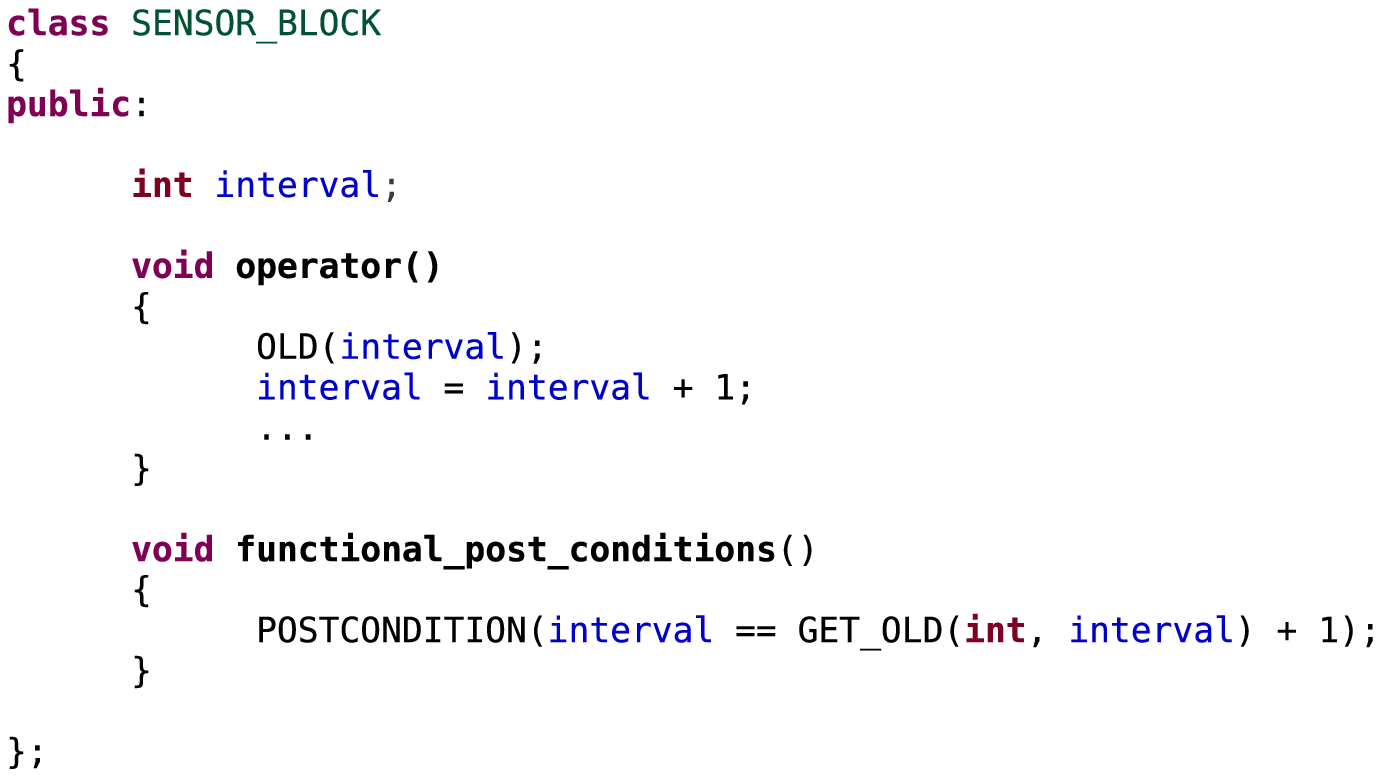} 
\caption{Functional contract for monitoring the state of a class attribute for Sensor from the running example}
\label{fig:funct_sens} 
\end{figure}
                     
For the real-time specifications, the properties supported by our framework are related to:
\begin{enumerate}
\setlength{\itemsep}{-1.5pt}
\item Execution times of the function blocks.
\item Cycle times of the control applications: our target applications are cyclic in nature and since they are safety critical, they must not violate the desired cycle time.
\item Jitter margin: applications must not exceed the maximum deviation from the cycle time that can be tolerated. This upper limit is called \textit{jitter margin}. 
\end{enumerate}
The cycle time and jitter margin for industrial component-based systems are usually predefined by the control engineer. The execution time analysis of the individual components is very difficult due to several reasons.  First, developing a comprehensive static analyzer for WCET analysis which can be used for any component-based system is very difficult. Second, components are often bought from vendors and used off-the-shelf due to which their source code may not be available. Third, some existing probabilistic approaches try to model the execution times of all components using a particular distribution function; we executed $24$ different function blocks from $11$ control applications and found that upon executing the blocks for 1000 cycles or more, the distributions converge to a single peak but are very different from one another. As a result, it is not possible to fix a standard probability density function a-priori to model the real-time behavior of all the function blocks. To overcomes these challenges, instead of using static analysis or a particular probability distribution function, we statistically \textit{estimate} the execution time upper bounds and specify contracts using these estimates. An example of a real-time contract for a component in our approach reads like ``Execution time of component $X$ should lie within $\gamma$ standard deviations from the mean execution time of $X$", where the execution time upper bound is computed as a function of the estimates of the mean and the standard deviation. Figure~\ref{fig:rt_sens} shows this contract for Sensor.  Our approach can be used with any component-based system and its performance is independent of the size of the application (in terms of the number of components).
\begin{figure} [htb]
\centering 
\includegraphics[width=.6\columnwidth]{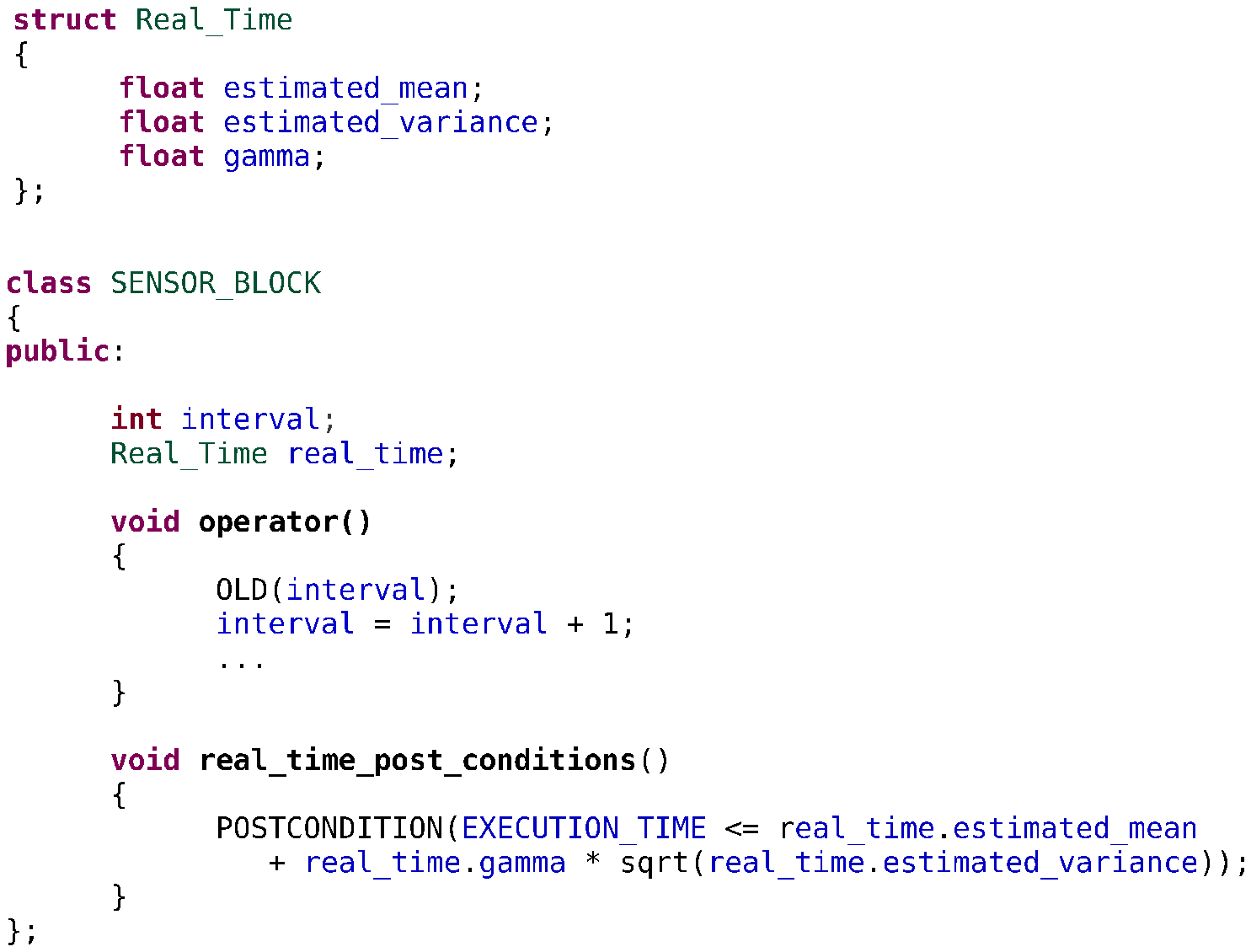} 
\caption{Stochastic real-time contract for Sensor from the running example}
\label{fig:rt_sens} 
\end{figure}
  
\section{Stochastic Contracts}\label{Section3}
We use statistical inference for analysing the execution times of the function blocks and specifying related contracts: we estimate the WCET values based on the empirical cumulative distribution functions (cdfs) of the execution times. 
By definition, the cdf of a random variable X is given by:
\begin{equation}\label{eq:cdf}
F(x)=P[X\leqslant x]
\end{equation} 
where $P[X\leqslant x]$ is the probability that the random variable X has value less than or equal to $x$. Let $\mu$ and $\sigma$ represent the mean and standard deviation of the unknown probability density of the execution time of a function block. Using equation~\eqref{eq:cdf} we can find an upper bound $\tau_\pi$ on the execution time such that $P[X \leq \tau_\pi]=\pi$ where $\pi \in [0, 1]$.
Now, $\tau_\pi$ can be expressed as a function of $\mu$ and $\sigma$ as:
\begin{equation}
\tau_\pi = \mu + \gamma \sigma \label{eq:pop}
\end{equation}  
where $\gamma$ gives us the distance of ${\tau_\pi}$ from the mean, $\mu$, in terms of number of standard deviations, $\sigma${\scriptsize{s}}. 
Every function block is executed for $n$ cycles and the execution times are recorded. Let $e_1, e_2, e_3...e_n$ denote this sample. The sample data is used for determining the empirical cdf. Figure~\ref{fig:energy_pack_cdf} shows the empirical cdf of our Sensor, based on 1000 cycles. Next, the mean, $\mu_n$, and the standard deviation, $s_n$ of the sample are computed as: \\
\begin{equation*}
{\mu}_n =  \frac{\sum_{i=1}^{n} e_i}{n} 				
\end{equation*}
\begin{equation*}
s_n=\sqrt{\frac{\sum^n_{i=1}{(e_i-{\mu}_n)}^2}{n-1}}
\end{equation*}

\begin{figure}
\centering 
\includegraphics[width=.6\columnwidth]{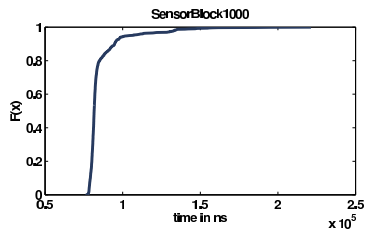} 
\caption[Empirical cumulative distribution functions]{Empirical cumulative distribution function of the Sensor from the running example based on sample size of 1000}
\label{fig:energy_pack_cdf} 
\end{figure}

Using equation~\eqref{eq:cdf} our framework \textit{estimates} an upper bound, $\widehat{\tau_\pi}$ on the execution time, if a threshold probability $0 \leq \pi \leq 1$ is given  by the application developer. For example, if the developer sets $\pi$ to 0.99, then $\widehat{\tau_\pi}$ is computed so that 99\% of the execution times lies within $\widehat{\tau_\pi}$ from the $(n+1)^{th}$ cycle onward. Thus, the accuracy of the bound depends on $\pi$.
Similar to equation~\eqref{eq:pop}, $\widehat{\tau_\pi}$ can be expressed as a function of $\mu_n$ and  $s_n$ as:
\begin{equation}
\widehat{\tau_\pi} = \mu_{n} + \gamma s_{n} \label{eq:sample}
\end{equation} 
The reason why we express the bound as a function of $\mu_n$ and $s_n$ is that it makes dynamic updating of the bound (as explained later) faster. The framework computes $\gamma$ only once and treats it as a constant while updating $\mu_n$ and $s_n$ and re-estimating the bound. 
$\widehat{\tau_\pi}$ , $\mu_n$  and $s_n$ are \textit{unbiased estimators}~\cite{hogg2005introduction} of $\tau_\pi$, $\mu$ and $\sigma$ respectively. This means that the expected value of $\widehat{\tau_\pi}$ is equal to $\tau_\pi$ and similarly for the other two parameters. Following is the formal definition of this notion.

\begin{defn}
Let $e_1, e_2, e_3...e_n$ be a random sample drawn from a population and let $\theta$ be an unknown parameter of the population which is to be estimated. Let $\widehat{\theta}=\upsilon(e_1, e_2, e_3...e_n)$ be a function of the sample. By definition, $\widehat{\theta}$ is an unbiased estimator of $\theta$ if the expected value of $\widehat{\theta}$ is equal to $\theta$, i.e. 
\hspace*{3.5cm}$E[\widehat{\theta}]=\theta $\qed
\end{defn}
We use unbiased estimators because they have desirable properties~\cite{hogg2005introduction}. The proof of the unbiasedness of $\mu_n$ and $s_n$ are standard results in statistical inference. They are also provided in \cite{EPFL-STUDENT-203686}. We use them in the following theorem.

\begin{mydef}
For any given function block, $\widehat{\tau_\pi}$ is an unbiased estimator of the upper bound ${\tau_\pi}$ on its execution time.
\end{mydef}

\begin{proof}
$E[\widehat{\tau_\pi}]=E[\mu_{n} + \gamma s_{n}]$\\
$=E[\mu_n]+E[\gamma s_n]$\\
$=E[\mu_n]+\gamma E[s_n]$\\
$=\mu + \gamma \sigma$ \\
$=\tau_\pi$  using equation~\eqref{eq:pop} \qedhere
\end{proof}

Algorithm~\ref{alg:alg3} describes the process for computing $\gamma$. It takes the execution time data and $\pi$ as inputs and returns $\gamma$. The algorithm sorts and bins the execution times, where the number of bins is the square root of the total data count and calculates the upper limit for the first bin (line 12). Then it computes the probability of the data falling into the first bin  (lines 13-20). If this probability exceeds (or is equal to) $\pi$, $\widehat{\tau_\pi}$ is set to the upper limit of the first bin. Otherwise, it computes the next bin's upper limit (line 25) and the process repeats until $\pi$ is reached. Ultimately, the algorithm returns $\gamma$ computed using equation~\eqref{eq:sample}.

\begin{algorithm}[htb]
\caption{Empirical Cdf}
\label{alg:alg3}

\begin{algorithmic}[1]
 \Procedure{EmpiricalCdf}{\textit{execution\_times}, $\pi$}: $\gamma$
 \small\State $\mathit{sum, y \gets}$ 0, 0
 \small\State $\mathit{buffer\_size \gets\textsc{size}(execution\_times)}$\\
 \small\hspace{0.5cm}$\mathit{\textsc{sort}(execution\_times)} $
 \small\State $\mathit{min \gets execution\_times[}$0$\mathit{]}$
 \small\State $\mathit{max \gets execution\_times[buffer\_size-}$1$\mathit{]}$
 \small\State $\mu_n\mathit{\gets \textsc{mean}(execution\_times)}$
 \small\State $s_n\mathit{\gets \textsc{std_dev}(execution\_times)}$
 \small\State $\mathit{upper\_bound \gets max}$
 \small\State $\mathit{number\_of\_bins \gets \textsc{sqrt}(buffer\_size)}$
 \State $\mathit{bin\_width \gets \frac{(max-min)}{number\_of\_bins}}$
 \small\State $\mathit{upper\_limit\_of\_bin \gets (execution\_times[}$0$\mathit{]+bin\_width)}$
 \small\While {$\mathit{upper\_limit\_of\_bin \leq max}$}
 \small\While{$\mathit{y \leq buffer\_size}$}
 \small\If{$\mathit{execution\_times[y] \leq upper\_limit\_of\_bin}$}
 \small\State $\mathit{sum \gets sum +}$1$\mathit{} $ \label{op2}
 \small\EndIf
 \small\State $\mathit{y \gets y+}$1
 \small\EndWhile
 \State $\mathit{probability \gets \frac{sum}{buffer\_size}}$
 \small\If{$\mathit{probability \geq} \pi$}
 \small\State $\widehat{\tau_\pi} \mathit{\gets upper\_limit\_of\_bin}$ 
 \small\State \textbf{break}
 \small\EndIf
 \small\State $\mathit{upper\_limit\_of\_bin \gets upper\_limit\_of\_bin+bin\_width}$
 \small\State $\mathit{sum, y \gets}$ 0, 0
 \small\EndWhile
 \State $\gamma \gets \frac{\widehat{\tau_\pi}-\mu_n}{s_n}$\\
 \small\hspace{0.5cm}\textbf{return} $\gamma$
 \EndProcedure
\end{algorithmic}
\end{algorithm}
Once $\gamma$ (based on the data from the first $n$ cycles) is obtained, we can specify contracts like ``Execution time of function block $X$ should not exceed $\gamma$ standard deviation from the mean of $X$" from the $\mathit{(n+1)^{th}}$ cycle onward. 
The parameter $\gamma$ is a characteristic property of a particular function block. This technique of computing a function block specific parameter $\gamma$ and using it for specifying the real-time contracts eliminates the problems of using the same probability density function to model the execution times of all function blocks. 

As the function blocks continue to execute beyond the first $n$ cycles, the accuracy of the upper bound may deteriorate due to several reasons. First, a function block may take longer time to execute due to platform related reasons such as waiting to acquire a resource and this might affect its execution time. Second, the execution time depends on the path of execution followed by a function block: in a certain cycle, it may have to handle an exception which might take longer than usual. Third, the execution time of a function block depends on interactions with the other function blocks and if the other function blocks have a delay, then that will affect the execution time of all the interacting function blocks. To account for these factors, it is important to dynamically update the computed upper bound from time to time. Since all computations are performed at runtime, running Algorithm~\ref{alg:alg3} too many times during the execution of the control application would degrade the application's performance. To prevent this, we only re-estimate $\mu_n$ and $s_n$ at runtime and update $\widehat{\tau_\pi}$ using the new value of $\mu_n$ and $s_n$ and the previously computed value of $\gamma$ (based on the first $n$ cycles). For updating the estimates, we use a sliding window mechanism: after a certain number $h$ of cycles, the oldest $h$ values of the execution time of a function block are replaced by the newest $h$ values and $\mu_n$ and $s_n$ are recomputed.

\section{Implementation}
We implemented the contract framework as a C++ library. Contracts are checked dynamically. When a contract fails, a message is logged. The logging is done by the FASA scheduler during the time left in each cycle after all function blocks are executed. This time is called the \textit{slack time}. If there is no slack time remaining in a cycle, the messages are stored in a buffer and carried on to the next cycle, where they are logged before the cycle's own messages. The contracts are specified within the function blocks in ``contract methods". This can be seen in the examples in Figure~\ref{fig:funct_sens} and Figure~\ref{fig:rt_sens}. To test the performance of our framework, we fixed an upper bound on the execution time overhead added due to contracts as 10\%, consistent with standard practices~\cite{Reference47}.  

\section{Validation and Results}
\label{Section6}
We evaluated our framework on five industrial case studies developed at ABB Corporate Research. The applications were tested on MacMini computers with 4 GB RAM and quad code processors, each core running at 1.8 GHz. The operating system was 64 bit Ubuntu 13.04.
The results are shown for 10000 executions of each application. For the real-time contracts, the cycle time was taken as 10ms for the first four case studies and 100ms for the fifth case study in order to handle network communication delays and the jitter margin was set to 0.1ms for all case studies. The entire list of contracts for each application is available in \cite{EPFL-STUDENT-203686}.
\begin{table*}[!htpb]
\caption[Performance summary of case studies]{Performance summary for five case studies}
\label{tab:perf_sum}
\centering
\begin{tabular}{c|c|c|c|c} \hline
Case Study & Number of contracts & \multicolumn{2}{|c|}{execution time} & overhead  \\ \cline{3-4}
&  & without contracts & with contracts &   \\ \hline
Simple Counter & 9 & 0.5447 & 0.5574 & 2.33\% \\\hline
Gaussian Generator & 27 & 1.4202  & 1.5161 & 6.75\% \\ \hline
Energy Pack Core Model$^{*}$ & 36 &  1.7467 & 1.8657 & 6.81\% \\ \hline
Binary Search$^{*}$ & 43 & 1.5865 & 1.6761 & 5.65\% \\ \hline
Net Proxy: sender & 10 & 0.4746 & 0.6359 & 33.99\% \\ \hline
Net Proxy: receiver & 10 & 0.8016 & 1.1531 & 43.85\% \\ \hline
\multicolumn{4}{c}{\small{*median execution time is computed due to outliers in the data.}}
  \end{tabular}
\end{table*}

\textbf{Simple Counter.} There is a single function block. It does not have any ports. An integer variable \texttt{counter} is initialized to 1 in the block constructor and is incremented in every cycle up to 10.  
For functional contracts, the framework checked a loop variant, a loop invariant and a postcondition to ensure the correct incrementation of \texttt{counter} up to 10 in every cycle. For the stochastic contracts, the value of the threshold probability $\pi$ was set to $0.95$ and the sliding window updating interval $h$ was set to 1.  

\textbf{Gaussian Generator.} The application generates random numbers having a Gaussian distribution. A Random Generator block generates two random numbers according to Uniform Distribution, $U[0,1]$. It then sends the two values to a Gaussian Generator block which generates two standard normal, $N(0,1)$ random values using the Box-Muller transformation\footnote{http://en.wikipedia.org/wiki/Box-Muller_transform}. These two values are sent to a Range Calculator which computes the range of the two Gaussian random numbers. The framework checked functional preconditions for verifying the connectivity of the ports for every function block before they started transmitting data. For the Random Generator, it checked postconditions to ensure that the generated random numbers lie within $[0,1]$, as required by the Gaussian Generator. For the Range Calculator, a postcondition ensured that the computed range is $\geq 0$.
The values of $\pi$ and $h$ were set to $0.97$ and 10 respectively for the stochastic properties. 

\textbf{Energy Pack Core Model.} This is our running example, shown in Figure~\ref{fig:runningeg}. The application is described in section~\ref{subsec:mtvtn}. The functional contracts included postconditions for monitoring the states of variables and attibutes, preconditions for ensuring the connectivity of ports and class invariants. For this application, $\pi$ was set to $0.95$ and $h$ to 5.  

\textbf{Binary Search.} A Random Integer Generator generates a random integer which is sent to an Array Creator. The latter stores the integer in a dynamic array. It sends the array to a Sorter to be sorted in ascending order. The sorted array is sent to a Binary Search block. This block has second input port where it receives a random number from the Random Integer Generator and searches for this random number in the sorted array that it receives from the Sorter and shows the index of the number if it is found, otherwise it prints $-1$ as the index. After every $10^{th}$ cycle, the memory allocated to the array is cleared. Thus, the maximum size of the array is $10$.  
The functional contracts checked preconditions for the connectivity of the ports and loop variants and invariants to ensure the total correctness of the loops in the Array Creator, the Sorter and the Binary Search block. For the execution time related contracts, $\pi$ was set to $0.98$ and $h$ to $5$.
 
\textbf{Net Proxy.} There are two separate applications launched on two separate host controllers. The first application has  a Sender and a Net Proxy Send block. The latter transmits integer values from the Sender through the network. The second application has a Net Proxy Receive block and a Receiver. The Net Proxy Receive block receives the integer value from the network and sends it to the Receiver. In order to synchronize the clocks on the two hosts, the Precision Time Protocol (PTP) is used. 
For the Sender, the functional contracts included a precondition to check the connectivity of the output port, a postcondition to monitor the state of an attribute and a class invariant. For the Receiver, a precondition checked the connectivity of the input port and a postcondition monitored the state of an attribute. $\pi$ was set to $0.99$ and $h$ to 10 for both applications. A real-time postcondition checked the compeletion time of the Receiver with respect to the Sender in a contract like ``Execution of Receiver must terminate within 0.2ms from the start of the execution of the Sender", in order to account for the network communication delays. 

Table~\ref{tab:perf_sum} shows the total number of the contracts, the mean execution times in ms for the five case studies with and without contracts enabled and the execution time overheads. The number of contracts include all functional and real-time contracts for every function block in the respective application~\cite{EPFL-STUDENT-203686}. The execution time overhead for the first four applications is well under our 10\% limit. For the fifth case study, the overhead was more due to delays in the network communication (some of the real-time contracts for the \textit{Receiver} relied on data obtained through the network). 

\section{Related Work}
\label{Section2}
\cite{Reference32} and \cite{Reference30} describe use of contracts for verification of programs in Ada. In \cite{Reference56} a tool has been developed for the BIP~\cite{Reference54} component framework which checks a program against programmer written specifications at runtime. \cite{Reference64} introduces a framework  for specifying contracts in temporal logic. The use of contracts for distributed embedded system design is shown in \cite{Reference23}. None of the above work supports real-time contracts for reasoning about metric time. 
\begin{table}[!htbp]
\scriptsize
\caption[Feature analysis of contract frameworks]{Characteristics of different contract frameworks}
\label{tab:lit_ana}
\centering
\tabcolsep=0.11cm
\begin{tabular}{p{3cm}p{1.5cm}p{1.5cm}p{1.5cm}}
\multicolumn{1}{p{2cm}}{} & {functional} & temporal & stochastic updates \\\hline
AdaCore\cite{Reference32} & \checkmark & $\times$ & $\times$\\\hline
Hi-Lite\cite{Reference30} & \checkmark & $\times$ & $\times$\\\hline
Barbacci et al.\cite{Reference8} & \checkmark & \checkmark & $\times$\\\hline
H\"{a}rtig et al.\cite{Reference24} & \checkmark & \checkmark& $\times$\\\hline
Stierand et al.\cite{Reference23} & $\times$ & \checkmark  & $\times$\\\hline
Sangiovanni-Vincentelli et al.\cite{Reference25} & \checkmark & $\times$  & $\times$\\\hline
RV-BIP\cite{Reference56} & \checkmark & $\times$ & $\times$ \\\hline
OCRA\cite{Reference64}& \checkmark & $\times$ &  $\times$\\\hline
Sojka et al.\cite{Reference28} & $\times$ & $\times$ & $\times$ \\\hline
\textbf{Stochastic Contracts} & \checkmark & \checkmark & \checkmark\\\hline
\end{tabular}
\end{table}
Closest to our idea are the works \cite{Reference24} and \cite{Reference8} in the sense that both focus on specification of functional and timed properties for real-time systems, but, unlike their tools, our tool does not rely only on statically defined contracts. Instead, it updates the contracts at runtime to ensure that they remain meaningful throughout the execution of an application, taking into consideration the changes in the underlying platform related factors. 
\cite{Reference28} and \cite{Reference25} show contracts in multiple layers of real-time systems. The similarity with our work lies in the layered structure of the contracts. In our framework, we specify real-time contracts in two hierarchical layers: function block level (WCET related) and application level (cycle time and jitter related). \cite{Reference25} only supports functional properties while \cite{Reference28} supports properties related to resource reservation. 

To our knowledge, our framework is the first to make use of statistical techniques for dynamically estimating and updating real-time contracts related to execution times of an application's components. 
Table~\ref{tab:lit_ana} summarizes the characteristic features of the existing contract frameworks which are relevant for our research and the characteristics of our own contract framework (the final row). 

\section{Conclusions and Future Work}
\label{Section7}
We presented a statistical inference based approach for computing real-time contracts for component-based real-time control applications. 
Based on experiments with industrial control applications deployed on a single controller (without network communication), we demonstrated that our contract framework adds an acceptable execution time overhead (much less than 10\%).

We used statistical inference for estimating upper bounds on the execution times of the function blocks. A possible extension would be to use machine learning techniques such as neural networks and compare the quality of the obtained estimates from the two approaches. The challenge in the latter approach lies in limiting the overhead added by the heavy computations at runtime, either by performing the learning offline or by other optimization techniques. 

\section{Acknowledgements}
We would like to thank Viktor Kuncak for his guidance and Carlo A.\ Furia, Chris Poskitt and Nadia Polikarpova for their feedback on the drafts of the paper. The first author is currently funded by ERC grant no.\ 291389.
%
\bibliographystyle{abbrv}
\bibliography{bibliography}  
%
%

\end{document}